\documentclass[12pt]{amsart}

\usepackage{amsthm,amssymb,amsmath,amsfonts}
\usepackage[shortlabels]{enumitem}
\usepackage{color}
\usepackage{graphicx}
\usepackage{caption}
\usepackage{subcaption}
\usepackage[margin=1.2in]{geometry}
\usepackage{hyperref}

\newtheorem{theorem}{Theorem}
\newtheorem{proposition}{Proposition}
\newtheorem{lemma}{Lemma}
\newtheorem{corollary}{Corollary}
\newtheorem{definition}{Definition}

\title{Skateboard Tricks and Topological Flips}
\author[J. Carlisle]{Justus Carlisle}
\author[K. Hammer]{Kyle Hammer}
\author[R. Hingtgen]{Robert Hingtgen}
\author[G. Martins]{Gabriel Martins}

\begin{document}


\begin{abstract}
We study the motion of skateboard flip tricks by modeling them as continuous curves in the group $SO(3)$ of special orthogonal matrices. We show that up to continuous deformation there are only four flip tricks. The proof relies on an analysis of the lift of such curves to the unit 3-sphere. We also derive explicit formulas for a number of tricks and continuous deformations between them.
\end{abstract}

\maketitle

\tableofcontents

\section{Introduction}
\label{sec:intro}

Our work is inspired by the following question. Consider two different skateboard flip tricks, call them $F_1$ and $F_2$. Is it possible for a skater to perform a sequence of many tricks (say one hundred tricks for example) such that every consecutive trick is very similar to each other, and such that the first trick in the sequence is $F_1$ while the last trick is $F_2$? If such a sequence exists we will say that the trick $F_1$ can be continuously deformed into the trick $F_2$.

In order to study this problem we will model a flip trick as a continuous curve in the group $SO(3)$ of special orthogonal matrices
\[SO(3) = \{ R \in M_{3\times 3}(\mathbb{R}) \mid R^TR = I_3, \quad \det(R)=1 \} \]
above $M_{3\times 3}(\mathbb{R})$ denotes the set of $3 \times 3$ matrices with real coefficients, $R^T$ denotes the transpose of the matrix $R$ and $I_3$ denotes the $3\times 3$ identity matrix.

More formally we will say that two tricks can be continuously deformed into each other if the curves modeling them are \emph{homotopic}.

Before establishing the formal framework, we would like to note that in our study of this problem we will also consider non-physical tricks. In reality, a skateboard trick would be modeled by a curve in $SO(3)$ of class $C^2$ which satisfies \emph{Euler's equations} of rigid body dynamics, or equivalently, by a locally length-minimizing curve in $SO(3)$ with respect to a certain metric in $SO(3)$ obtained from the inertia tensor of the skateboard, as described by Arnold in \cite{Ar, Ar2}. Other great references for concepts of geometric mechanics can be found in \cite{AM, HSS}.

We will consider a broader class of curves in $SO(3)$. We require our curves to be continuous, but not necessarily satisfy any particular differential equation. We will call these curves \emph{topological flips}, as they are a generalization of a physical flip trick. You could think of these tricks as movements that could be realized by a skateboard equipped with perfect jet propulsors, which would allow one to control their movement in the air completely.

Allowing for this more general class of curves makes the problem a lot more flexible and tractable. The disadvantage is that it does not provide a clear answer for the more difficult problem regarding only physical tricks. We will come back to this point later in the text.

\section{Formal Framework}
\label{sec:framework}

We now define the resting position for the skateboard in 3-space. The origin of our coordinate system will be placed at the center of mass of the skateboard. The $z$-axis points upwards and is perpendicular to the plane defined by the center of the board. The $y$-axis points to the front of the board (also called the nose of the board), in that way the $yz$-plane divides the skateboard into two equal parts. Finally, if you align your body with the $z$-axis while facing the positive $y$-axis, the $x$-axis points to the right side of the $y$-axis. The skater on the board will be facing the positive $x$-axis with their right foot on the back of the board (also called the tail of the board). We note that with this choice, the coordinate axes agree with the principal axes of inertia of the skateboard, so that in these coordinates the inertia matrix would be diagonalized.

We then consider an orthonormal frame $\vec{v}_1, \vec{v}_2, \vec{v}_3$ which is locked onto the skateboard and based at its center of mass. As the skateboard moves about 3-space during a flip trick, this frame moves along with it. We obtain in this way a time-dependent frame $\vec{v}_1(t), \vec{v}_2(t), \vec{v}_3(t)$ which will describe the rotational movement of the skateboard. We collect this frame into a $3 \times 3$ matrix having the vectors in this moving frame as columns:
\[R(t) = 	\begin{bmatrix}
		\uparrow & \uparrow & \uparrow \\
		\vec{v}_1(t) & \vec{v}_2(t) & \vec{v}_3(t) \\
		\downarrow & \downarrow & \downarrow
		\end{bmatrix}\]

We will assume that the trick will be performed from time $t=0$ to $t=1$. At time $t=0$ the skateboard is at its resting position and the frame is simply:
\[ \vec{v}_1 = (1,0,0), \quad \vec{v}_2 = (0,1,0), \quad \vec{v}_3 = (0,0,1) \]

\noindent in particular $R(0) = I_3$ is the $3 \times 3$ identity matrix.

Since this frame is orthonormal at every instant $t \in [0,1]$ the matrix $R(t)$ is always orthogonal, that is $R(t)^TR(t) = I_3$. Notice that by taking determinants in the previous equation we find that:
\[\det(R(t))^2 = \det \big( R(t)^T \big) \cdot \det\big( R(t) \big) = \det\left(R(t)^TR(t)\right) = \det(I_3) = 1\]
which implies that $\det(R(t)) = \pm 1$ for all $t \in [0,1]$. We will further assume that the motion of the skateboard is continuous, that is, we will assume that the function $R(t)$ is continuous in time. This means that every component of the matrix $R(t)$, or equivalently every component of the vectors $\vec{v}_i(t)$ in the frame, depend continuously on the time parameter $t$. This implies that $\det(R(t)) = 1$ for all $t \in [0,1]$, since $\det(R(0)) = \det(I_3) = 1$.

We have found that the rotational motion of flip tricks can be described as a continuous curve in the group $SO(3)$ of special orthogonal matrices. There is however one last condition that these continuous curves must satisfy. At the end of the flip trick the skateboard must land back on its wheels, aligned with the $y$-axis as it started.

We can see that there are only two configurations allowed for $R(1)$, either the skateboard lands back at its resting position, in which case $R(1) = I_3$, or it lands in a reversed position where the front of the skateboard (its nose) points towards the negative direction of the $y$-axis. This is the landing position obtained for example after the skateboard performs a 180 degree rotation about the $z$-axis. In this case one can see that $R(1)$ must be equal to the matrix $\mathcal{O}$ defined below:
\[\mathcal{O} = \begin{bmatrix}
-1	&	0	&	0	\\
0	&	-1	&	0	\\
0	&	0	&	1
\end{bmatrix}\]

We summarize this discussion with the following definition.

\begin{definition}
A topological flip is a continuous curve $R: [0,1] \to SO(3)$ such that $R(0)= I_3$ and $R(1) = I_3$ or $\mathcal{O}$.
\end{definition}

We are now able to discuss when two topological flips $R(t)$ and $F(t)$ can be continuously deformed into each other. Recall that a homotopy between $R(t)$ and $F(t)$ is a continuous function $H: [0,1] \times [0,1] \to SO(3)$ such that:

\begin{enumerate}[i)]
	\item $H(0,t) = R(t)$ for all $t \in [0,1]$
	\item $H(1,t) = F(t)$ for all $t \in [0,1]$
	\item $H(s,0) = I_3$ for all $s \in [0,1]$
	\item $H(s,1) = I_3$ or $\mathcal{O}$ for all $s \in [0,1]$
\end{enumerate}

Notice that for such a homotopy between $R(t)$ and $F(t)$ to exist it must be the case that $R(1) = F(1)$. In the case that such a homotopy exists between two topological flips we will say that these tricks are \emph{homotopic} and we will denote $R(t) \approx F(t)$.

We notice that the relation $\approx$ defines an equivalence relation on the set of topological flips. We call an equivalence class of this relation a \emph{homotopy class} and we denote the homotopy class of a given topological flip $R(t)$ by $[R(t)]$.

Our main result is the following.

\begin{theorem}[Four Tricks Theorem]
There are exactly four homotopy classes of topological flips. The set of homotopy classes has a natural group structure isomorphic to the cyclic group $\mathbb{Z}/4\mathbb{Z}$.
\end{theorem}

In the next section, we will carefully describe a number of skateboard flip tricks. The four homotopy classes can be represented by various different choices of tricks. The most obvious obvious choice is perhaps the one consisting of the iterations of a shove-it.

Below we list these four tricks with their corresponding residue class in $\mathbb{Z}/4\mathbb{Z}$.
\[\begin{array}{rcl}
0 & \leftrightarrow & \text{Ollie} \\
1 & \leftrightarrow & \text{180 Shove-it} \\
2 & \leftrightarrow & \text{360 Shove-it} \\
3 & \leftrightarrow & \text{540 Shove-it}
\end{array}\]

Another interesting choice of representatives is given by a combination of the shove-it and the kickflip.
\[\begin{array}{rcl}
0 & \leftrightarrow & \text{Ollie} \\
1 & \leftrightarrow & \text{180 Shove-it} \\
2 & \leftrightarrow & \text{Kickflip} \\
3 & \leftrightarrow & \text{Varial Kickflip}
\end{array}\]

\section{Explicit Formulas}
\label{sec:formulas}

We will now provide explicit formulas for a few standard skateboard tricks, recall that we have chosen the domain of the curves modeling the topological flips to be the interval $[0, 1]$.

The first trick we will define is the \textit{ollie}. It consists of a simple jump of the skater. We will represent it with the constant topological flip $O(t) = I_3$.

Next we analyze the trick called the \textit{shove-it} (also known as a backside 180 shove-it). This trick consists of a 180 rotation of the skateboard about the $z$-axis where the skater pushes the tail of the board backwards with their back foot. This means that the rotation will follow a \textit{left-handed} orientation. Since this is a rotation about the $z$-axis we have that the last vector in the orthonormal frame is stationary $\vec{v}_3(t) = (0, 0, 1)$, while the two first vectors $\vec{v}_1(t)$ and $\vec{v}_2(t)$ will perform a 180 degree rotation in the clockwise direction in the $xy$-plane. We will denote the curve in $SO(3)$ representing the shove-it by $S(t)$, it has the form:
\[S(t) = \begin{bmatrix}
\cos(\pi t)     & \sin(\pi t)   & 0 \\
-\sin(\pi t)    & \cos(\pi t)   & 0 \\
0               & 0             & 1
\end{bmatrix}\]

Since the matrix operations of multiplication and inversion are continuous, we are able to use them to derive formulas for new topological flips (which must always be continuous curves in $SO(3)$). For example, we can now easily write the formula for the trick called a \textit{360 shove-it}, which consists of a 360 degree rotation about the $z$-axis with a left-hand orientation. The curve in $SO(3)$ representing the 360 shove-it can be obtained by squaring $S(t)$:
\[S(t)^2 = \begin{bmatrix}
\cos(2\pi t)    & \sin(2\pi t)  & 0 \\
-\sin(2\pi t)   & \cos(2\pi t)  & 0 \\
0               & 0             & 1
\end{bmatrix}\]

Similarly a \textit{frontside shove-it} is a trick which consists of a 180 degree rotation of the skateboard about the $z$-axis, where in this case the skater pushes the tail of the board forwards with their back foot. This means that this rotation will follow a \textit{right-handed} orientation. The curve in $SO(3)$ representing the frontside shove-it can be obtained by inverting $S(t)$:
\[S(t)^{-1} = \begin{bmatrix}
\cos(\pi t) & -\sin(\pi t)  & 0 \\
\sin(\pi t) & \cos(\pi t)   & 0 \\
0           & 0             & 1
\end{bmatrix}\]

In a similar fashion we can express a 360 frontside shove-it as $S(t)^{-2}$, a 540 shove-it as $S(t)^3$, a 540 frontside shove-it as $S(t)^{-3}$ and so on.

We now move our attention to a different axis of rotation and describe the formula for the trick called a \textit{kickflip}. In this trick the skater rotates their skateboard by 360 degrees about the $y$ axis with a \textit{left-handed} orientation, that is the board spins counter-clockwise if viewed from the back. The curve $K(t)$ in $SO(3)$ representing the kickflip has the form:
\[K(t) = \begin{bmatrix}
\cos(2\pi t)    & 0 & -\sin(2\pi t) \\
0               & 1 & 0             \\
\sin(2\pi t)    & 0 & \cos(2\pi t)
\end{bmatrix}\]

In a double kickflip the skateboard rotates by 720 degrees about the $y$ axis with a left-handed orientation, the associated curve is then $K(t)^2$. A \textit{heelflip} is a trick where the skateboard rotates by 360 degrees about the $y$ axis with a \textit{right-handed} orientation, the associated curve is $K(t)^{-1}$.

All of the previous examples consist of rotations with constant angular velocity about a principal axis of inertia of the skateboard, therefore they are in fact physical flips in the sense that they satisfy Euler's equations of rigid body motion. The situation becomes more complicated once we combine motions about two axes simultaneously.

Consider for example the trick called a \textit{varial kickflip}. This trick's motion can be described as a simultaneous backside 180 shove-it together with a kickflip. We may create a curve representing this trick by using the multiplication operation and obtain $S(t)K(t)$. However, this representation of a varial kickflip has a number of issues. Firstly, the notion of a ``simultaneous motion of a kickflip and backside shove-it'' is not well defined, this is because the multiplication operation in $SO(3)$ is not commutative. For example the curve $K(t)S(t)$ can also be described as a simultaneous kickflip and backside shove-it, however its movement does not resemble anything that a skater might call a varial kickflip. Furthermore, the topological flip $S(t)K(t)$ is not physical in the sense that it does not satisfy Euler's equation of rigid body motion (and in fact this is true no matter the choice of moments of inertia for the skateboard). However visually the curve $S(t)K(t)$ looks very similar to a physical varial kickflip and so we will use it to model this trick. In the future we plan to study how close the curves obtained by multiplication are to the physical curves, but this question is beyond the scope of the present work. 

This caveat having been clarified, we can apply this multiplication technique to obtain formulas for many known tricks. A \textit{360 flip} consists of a simultaneous 360 backside shove-it and kickflip, it can be described as $S(t)^{2}K(t)$. A \textit{varial heelflip} is a simultaneous heelflip and frontside shove-it, it can be described by $S(t)^{-1}K(t)^{-1}$.

As a last example we will describe the \textit{hardflip}. This is a simultaneous half kickflip (a 180 degree rotation about the axis joining the tail to the nose of the skateboard) and a 180 degree rotation about the $x$-axis in the right-hand orientation. If we denote the rotation about the $x$-axis by:
\[U(t) = \begin{bmatrix}
1 & 0           & 0             \\
0 & \cos(\pi t) & -\sin(\pi t)  \\
0 & \sin(\pi t) & \cos(\pi t)
\end{bmatrix}\]

\noindent then we may describe the hardflip as $U(t)K(t/2)$.

\section{Quaternions and Rotations}

Recall that the quaternion numbers $\mathbb{H}$ can be represented in the form
\[
q = q_{0} + q_{1}\textbf{i} + q_{2}\textbf{j} + q_{3}\textbf{k}
\]
where the numbers $\textbf{i},\textbf{j}$ and $\textbf{k}$ satisfy the relations:
\[
\textbf{i}^{2} = \textbf{j}^{2} = \textbf{k}^{2} = -1, \ \textbf{i}\textbf{j}\textbf{k} = -1
\]
The real number $q_0$ is called the scalar part (also the real part) of $q$ and the number $q_{1}\textbf{i} + q_{2}\textbf{j} + q_{3}\textbf{k}$ is called the vector part (also the imaginary part) of $q$. We will often think of a quaternion as a 4-vector $(q_0, q_1, q_2, q_3)$ and identify the vector part of a quartenion with the 3-vector $(q_1, q_2, q_3)$.

The conjugate $\overline{q}$ of a quaternion $q$ is given by switching the sign of its vector part:
\[
\overline{q} = q_{0} - q_{1}\textbf{i} - q_{2}\textbf{j} - q_{3}\textbf{k}
\]

In general, conjugation reverses the order of quaternionic multiplication:
\[
\overline{pq} = \overline{q} \cdot \overline{p}
\]

As with complex numbers, the real and imaginary parts of a quaternion can be written in terms of a quaternion and its conjugate, 
\[
\text{Re}(q) = \frac{q + \overline{q}}{2},  \hspace{1cm} \text{Im}(q) = \frac{q - \overline{q}}{2}.
\]

Another similarity between the complex and quaternion numbers is that we may write the modulus of a quaternion as a product with its conjugate,
\[
q\overline{q} = \overline{q}q = q_{0}^{2} + q_{1}^{2} + q_{2}^{2} + q_{3}^{2} = | q |^{2}.
\]
as a consequence the inverse of a nonzero quaternion is given by 
\[
q^{-1} = \frac{\overline{q}}{| q |^{2}}.
\]

Let us now consider the conjugation action of a fixed nonzero quaternion, $q \neq 0$, on purely imaginary quaternions $\mathbf{v}$ given by $q\mathbf{v}q^{-1}$. As $\mathbf{v}$ is a purely imaginary quaternion, we have that $\overline{\mathbf{v}} = -\mathbf{v}$. When computing the conjugate of $q\mathbf{v}q^{-1}$, we find that 
\[
\overline{q\mathbf{v}q^{-1}} = \overline{q^{-1}} \cdot \overline{\mathbf{v}} \cdot \overline{q} = \frac{q}{|q|^{2}} (-\mathbf{v}) \overline{q} = - q\mathbf{v}q^{-1}.
\]
With this it is clear that $\text{Re}(q\mathbf{v}q^{-1}) = 0$, that is $q\mathbf{v}q^{-1}$ is itself a purely imaginary quaternion.

Furthermore, notice that
\begin{align*}
    q(\mathbf{u} + \mathbf{v})q^{-1} &= q\mathbf{u}q^{-1} + q\mathbf{v}q^{-1} \\
    q(\alpha \mathbf{v})q^{-1} &= \alpha q\mathbf{v}q^{-1}
\end{align*}

\noindent for every purely imaginary quaternions $\mathbf{u}$, $\mathbf{v}$ and scalar $\alpha$.  It follows that the conjugation action is real linear. By identifying the purely imaginary quaternions with $\mathbb{R}^3$ we obtain a map taking nonzero quaternions to linear maps on $\mathbb{R}^{3}$. 
\[
\rho : \mathbb{H} - \{0\} \to L(\mathbb{R}^{3}), \ \ \rho(q) \mathbf{v} = q\mathbf{v}q^{-1}
\]

Notice that because $q^{-1} = \overline{q}/|q|$, the components of the matrix associated to $\rho(q)$ can be expressed as a quadratic polynomial on $q$ divided by the function $|q|$, so that $\rho$ is a continuous map. Furthermore $\rho(qp) = \rho(q)\rho(p)$, so that $\rho$ is a continuous group homomorphism.

A unit quaternion is a quaternion with unit length, i.e $|q|^2=q\overline{q} = 1$. Geometrically the set of unit quaternions is identified with the unit 3-sphere $S^{3}$ inside of $\mathbb{R}^4$. Given a unit quaternion $q$ and an imaginary quaternion $\mathbf{v}$, notice that:
\[ | \rho(q) \mathbf{v} | = |q\mathbf{v}q^{-1}| = |q||\mathbf{v}||q^{-1}| = |\mathbf{v}|\]

Since $\rho(q)$ leaves the norm of vectors invariant, it must preserve the inner product structure on $\mathbb{R}^{3}$ so that the matrix associated to $\rho(q)$ is orthogonal. Furthermore since $\rho$ is continuous and $\det(\rho(q)) = \pm 1$ for $q \in S^3$, it must be the case that $\det(\rho(q)) = 1$ for all $q$ in $S^3$, as $\rho(1) = I_3$. We may then restrict $\rho$ to the unit 3-sphere and obtain a map $\rho: S^3 \to SO(3)$.

Given a unit quaternion $q$, we may write $q = \cos\theta + \sin\theta \mathbf{u}$ for some unit imaginary quaternion $\mathbf{u}$. A somewhat long computation allows us to write the following \textit{Rodriguez formula} for a given imaginary quaternion $\mathbf{v}$:
\[ \rho(q) \mathbf{v} = q\mathbf{v}q^{-1} = \cos(2\theta)\mathbf{v}_\perp + \sin(2\theta) (\mathbf{u} \times \mathbf{v}) + \mathbf{v}_\parallel \]

\noindent where $\mathbf{v}_\parallel = (\mathbf{v} \cdot \mathbf{u}) \mathbf{u}$ and $\mathbf{v}_\perp = \mathbf{v} - \mathbf{v}_\parallel$.
This formula shows that $\rho(q)$ can be described geometrically as a \textit{right-handed} rotation of angle $2\theta$ about the axis of rotation $\mathbf{u}$. Since every special orthogonal transformation in 3-space is a rotation about a certain invariant axis, we obtain the following:

\begin{corollary}
The map $\rho: S^3 \to SO(3)$ is surjective.
\end{corollary}

For the sake of completeness, we will also mention that in the coordinates $q = q_{0} + q_{1}\textbf{i} + q_{2}\textbf{j} + q_{3}\textbf{k}$, the image of a unit quaternion $q$ under the map $\rho$ is given by the following formula:
\[
\rho(q) = \left[ \begin{array}{ccc} 
1 - 2(q_{2}^{2} + q_{3}^{2}) & 2(q_{1}q_{2} - q_{3}q_{0}) & 2(q_{1}q_{3} + q_{0}q_{2}) \\ 
2(q_{1}q_{2} + q_{3}q_{0}) & 1 - 2(q_{1}^{2} +q_{3}^{2}) & 2(q_{2}q_{3}-q_{1}q_{0}) \\
2(q_{1}q_{3}-q_{2}q_{0}) & 2(q_{2}q_{3} + q_{0}q_{1}) & 1 - 2(q_{1}^{2} + q_{2}^{2})
\end{array} \right]
\]

\section{Lifting Topological Flips to the 3-sphere}
\label{sec:lifting}

In this section we will show that the map $\rho: S^3 \to SO(3)$ is a covering map. This will allow us to lift topological flips and homotopies between them to the 3-sphere. We will then explore the simpler geometry of the sphere to obtain a few explicit homotopies between different topological flips.

Recall that a map $\pi: X \to B$ between topological spaces is called a covering map if for any point $b \in B$ there is an open neighborhood $V$ of $b$ such that its preimage $\pi^{-1}(V)$ is a union of disjoint open sets $\pi^{-1}(V) = \coprod_{\alpha \in A} U_\alpha$ and such that the restrictions of $\pi$ to the open sets $U_\alpha$ denoted by  $\pi\mid_{U_\alpha}:U_\alpha \to V$ are homeomorphisms. The space $B$ is called the \textit{base space} and $X$ is called the \textit{covering space}.

\begin{proposition}
The map $\rho: S^3 \to SO(3)$ is a covering map.
\end{proposition}

\begin{proof}

We will first establish a small lemma.

\begin{lemma}
Let $A$ be a matrix in $SO(3)$ such that $A = \rho(q)$. Then $\rho^{-1}(A) = \{q,-q\}$.
\end{lemma}

\begin{proof} Notice that if $p \in S^3$ is a different quaternion such that $\rho(p) = A$, then we know that for every imaginary quaternion $\mathbf{v}$ we have that:
\[ p\mathbf{v}p^{-1} = q\mathbf{v}q^{-1} \]

\noindent which is equivalent to the equation $q^{-1}p\mathbf{v} = \mathbf{v}q^{-1}p$. This means that the unit quaternion $q^{-1}p$ commutes with any imaginary quaternion, and in fact it must commute with every quaternion since real quaternions belong to the center of the quaternion algebra. Furthermore, since the center of the quaternion algebra is in fact equal to the real quaternions, we find that $q^{-1}p$ is real. Since $q^{-1}p$ is a unit quaternion we may conclude that $q^{-1}p=\pm 1$, hence $p = \pm q$ which establishes the lemma.
\end{proof}

We now show that $\rho$ must be an open map. Since $S^3$ is compact and $SO(3)$ is Hausdorff, $\rho$ is a closed map. Now given an open set $U \subset S^3$, let:
\[ \widehat{U} = \{q \in S^3 \mid q\in U \text{ or } -q\in U\}.\]

Then notice that the surjectivity of $\rho$ allows us to write:
\[ SO(3)-\rho(U) = SO(3)-\rho(\widehat{U}) = \rho(S^3 - \widehat{U}).\]

Since $\rho$ is a closed map, $\rho(S^3 - \widehat{U})$ must be closed, therefore $\rho(U)$ must be open.\\

Now let $A_0$ be any matrix in $SO(3)$ and let $\rho^{-1}(A_0)=\{q_0,-q_0\}$. Consider the neighborhoods:
\[\begin{aligned}
    U_+ & = \{ q \in S^3 \mid |q-q_0|<1 \} \\
    U_- & = \{ q \in S^3 \mid |q- (-q_0)|<1 \} \\
\end{aligned}\]

Notice that $q \in U_+ \Rightarrow -q \in U_-$ and $q \in U_- \Rightarrow -q \in U_+$. Let's then define the open neighborhood $V = \rho(U_+ \cup U_-) = \rho(U_\pm)$ of $A_0$. Clearly we have that $\rho^{-1}(V) = U_+ \cup U_-$, and since $\rho$ is an open map, $V$ is an open subset of $SO(3)$. It still remains to show that $U_+$ and $U_-$ are disjoint and that the restrictions of $\rho$ to $U_+$ and $U_-$ are homeomorphisms.\\

We now show that $U_+$ and $U_-$ are disjoint. Suppose by contradiction that there is a unit quaternion $q \in U_+ \cap U_-$, then notice that:
\[2 = 2|q| = |q - (-q)| = |q - q_0 + q_0 - (-q)| \leq |q - q_0| + |q_0 - (-q)|< 1 + 1 = 2 \]

\noindent which is a contradiction, so that $U_+$ and $U_-$ are disjoint.\\

Finally we consider the restrictions $\rho\mid_{U_\pm}:U_\pm \to V$ which are clearly continuous bijective functions. Since $\rho: S^3 \to SO(3)$ is an open map, and since $U_\pm$ are open sets, we may conclude that the restrictions $\rho\mid_{U_\pm}$ are also open maps and therefore homeomorphisms, which finishes our proof.
\end{proof}

There are many interesting theorems that relate continuous curves in the base space to continuous curves in the covering space. Let $\pi: X \to B$ be a covering map and $c: [a, b] \to B$ be a continuous curve in the base space. A \textit{lift} of $c$ is a curve $\gamma: [a, b] \to X$ in the covering space such that $\pi(\gamma(t)) = c(t)$. A crucial theorem regarding lifts of curves is the following (see \cite{Ha} Theorem 1.7 (a)).

\begin{theorem}
Let $c: [a, b] \to B$ be a continuous curve starting at a point $p$. Given a choice of point $x \in \pi^{-1}(p)$, there exists a unique lift $\gamma: [a, b] \to X$ of $c$ such that $\gamma(a) = x$.
\end{theorem}

This theorem allows us to relate topological flips to a certain family of curves in the 3-sphere $S^3$. Since $\rho^{-1}(I_3)=\{1,-1\}$ and $\rho^{-1}(\mathcal{O}) = \{\textbf{k}, -\textbf{k}\}$ we may conclude that the lift of a topological flip based at $1$ must be a continuous curve $\gamma: [0,1] \to S^3$ such that $\gamma(0)=1$ and $\gamma(1) \in \{1,\textbf{k},-1,-\textbf{k}\}$. We call such a curve in the 3-sphere a \textit{quaternionic flip}. We obtain a one-to-one correspondence:
\[\Big\{ \text{ Topological flips } \Big\} \leftrightarrow \Big\{ \text{ Quaternionic flips } \Big\}\]

Notice that $\{1,\mathbf{k},-1,-\mathbf{k}\}$ is the subgroup of $S^3$ generated by $\mathbf{k}$, we will denote it by $\langle \mathbf{k} \rangle$.

Another interesting theorem concerns the lift of homotopies in the base $B$ to homotopies in the covering space $X$ (see \cite{Ha} Theorem 1.7 (b)).

\begin{theorem}
Let $H:[0,1]\times[a,b] \to B$ be a homotopy in $B$ such that $H(s,a)=p$ for all $s$ in $[0,1]$. Given a choice of point $x \in \pi^{-1}(p)$ there is a unique homotopy $\tilde{H}:[0,1]\times[a,b] \to X$ such that $\pi (\tilde{H}(s,t)) =  H(s,t)$ and $\tilde{H}(s,a) = x$ for all $s$ in $[0,1]$.
\end{theorem}

A consequence of this theorem is that we may obtain a one-to-one correspondence between the set $[TF]$ of homotopy classes of topological flips and the set $[QF]$ of homotopy classes of quaternionic flips.
\[\Big\{ \text{ Topological flips } \Big\}\Big/\approx \;\leftrightarrow\; \Big\{ \text{ Quaternionic flips } \Big\}\Big/\approx\]

We will denote this ``lifting bijection'' by $\lambda: [TF] \to [QF]$. With only a few more results from algebraic topology we will soon be able to prove the four tricks theorem.\\

Let's now analyze the quaternionic lifts of the tricks described in section \ref{sec:formulas}. Recall that the shove-it $S(t)$ consists of a left-handed 180 degree rotation of the skateboard about the $z$-axis. By using the Rodriguez formula we may compute its quaternionic lift:
\[ \sigma(t) = \cos \left( \frac{\pi t}{2} \right) - \sin \left( \frac{\pi t}{2} \right)\textbf{k} \]

Since the $\mathbf{i}$ component of $\sigma$ vanishes, we may visualize the curve $\sigma(t)$ by plotting its image in the 3-dimensional subspace generated by $1$, $\mathbf{j}$ and $\mathbf{k}$. The intersection of the unit 3-sphere with this subspace will form a 2-dimensional sphere, see figure \ref{fig:shuvits}. For aesthetic reasons we will use the coordinates $x - y\mathbf{j} - z\mathbf{k}$ in our plots. Please notice the negative signs in the $y$ and $z$ coordinates.

\begin{figure}[ht]
    \vspace{-6pt}
	\centering
	\begin{subfigure}[t]{0.31\textwidth}
		\centering
		\includegraphics[width=\textwidth]{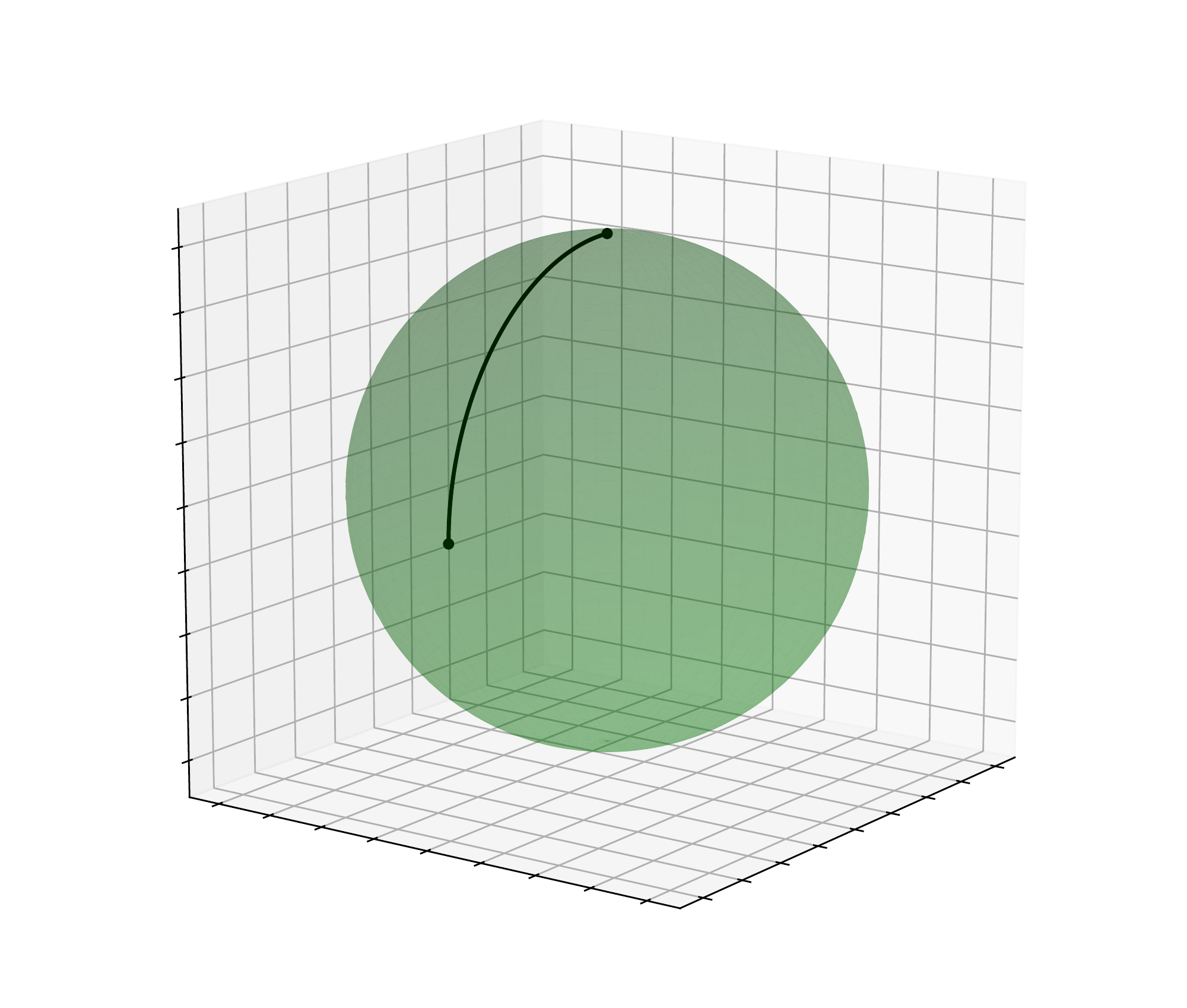}
		\caption{$\sigma$ curve}		
	\end{subfigure}
	\quad
	\begin{subfigure}[t]{0.31\textwidth}
		\centering
		\includegraphics[width=\textwidth]{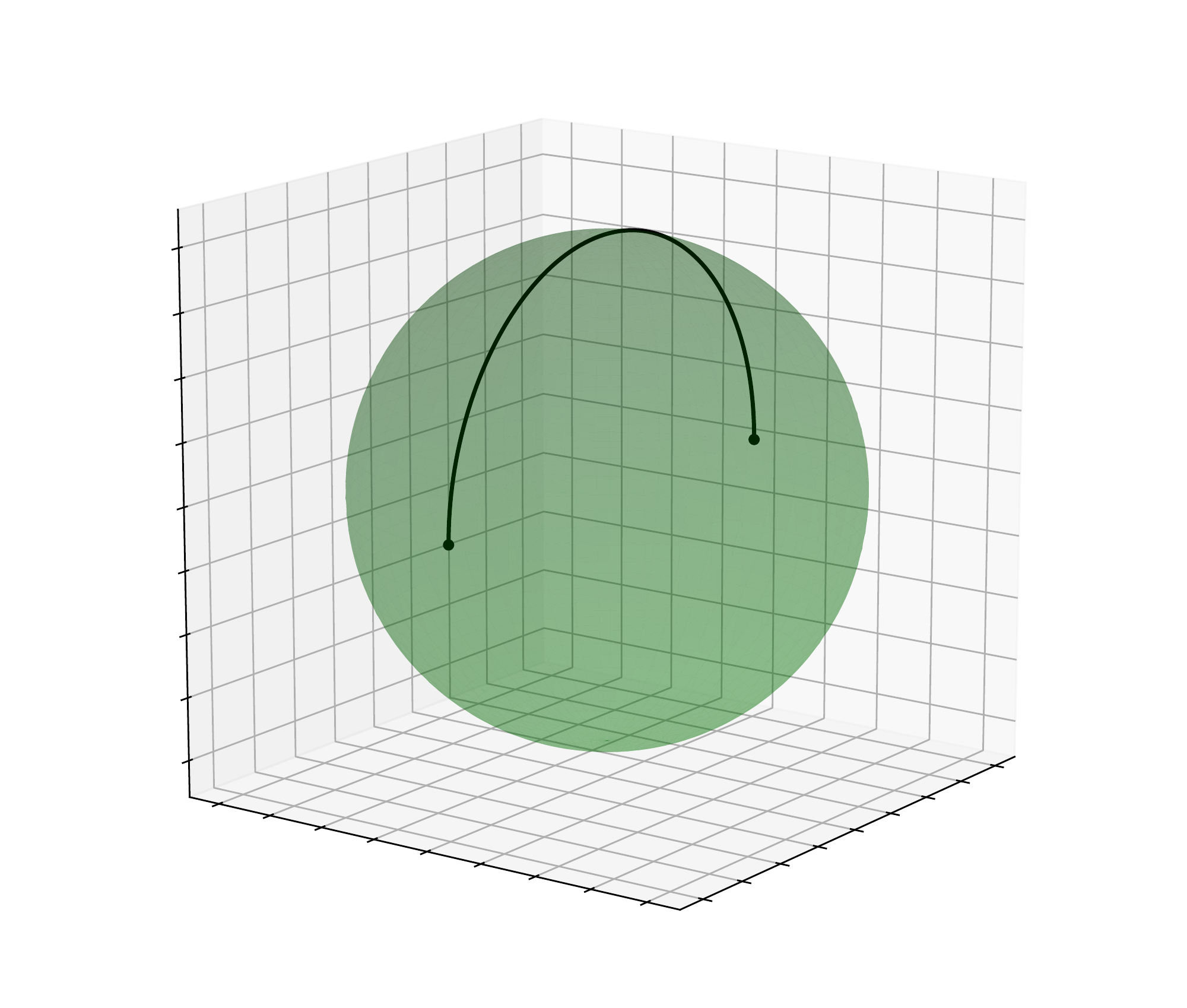}
		\caption{$\sigma^2$ curve}
	\end{subfigure}
	\quad
	\begin{subfigure}[t]{0.31\textwidth}
		\centering
		\includegraphics[width=\textwidth]{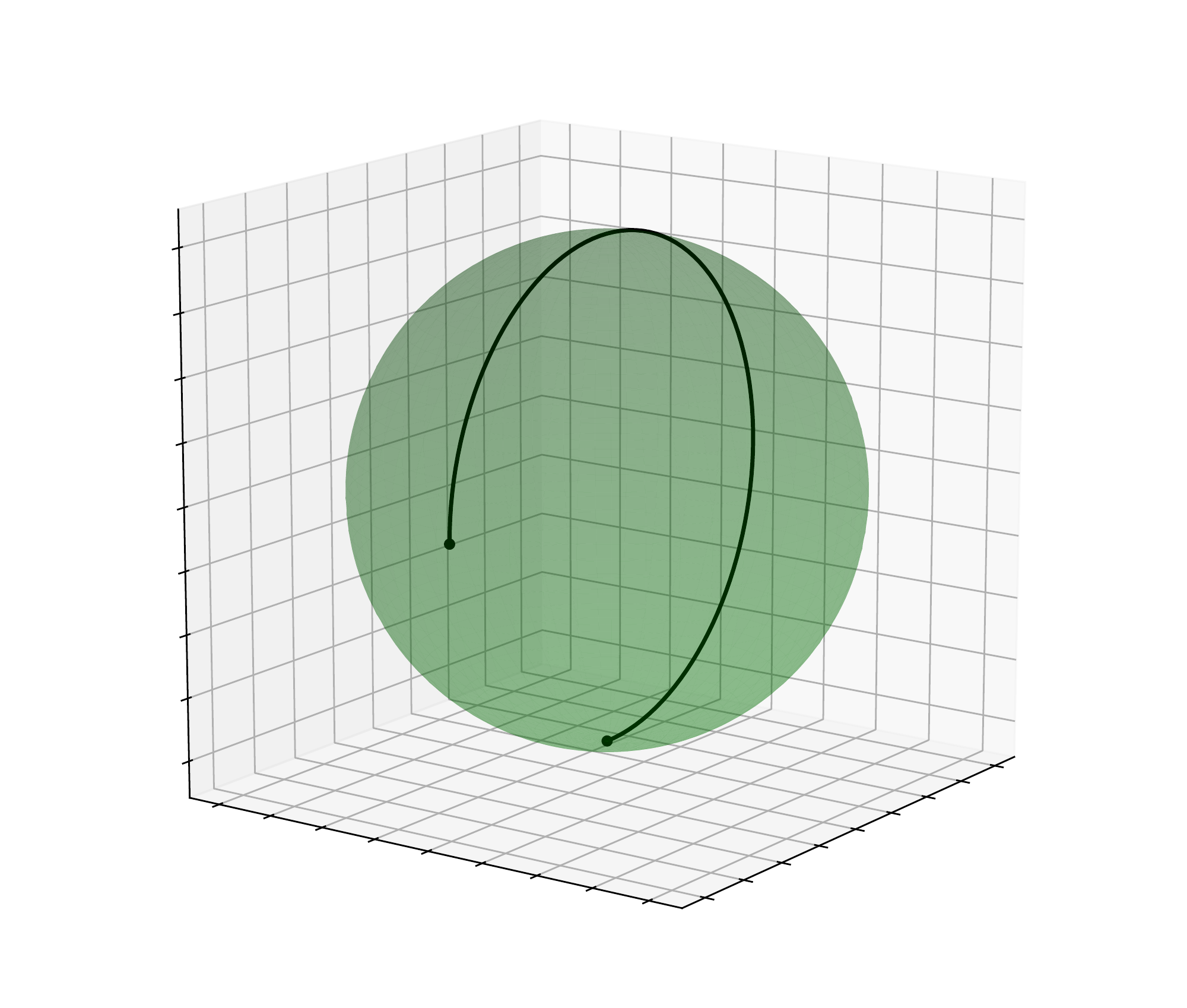}
		\caption{$\sigma^3$ curve}
	\end{subfigure}
	\caption{Quaternionic lifts of the powers of a shove-it}
	\label{fig:shuvits}
	\vspace{-4pt}
\end{figure}

We notice that $\sigma(1) = -\mathbf{k}$. The 360 shove-it can be lifted to the curve:
\[ \sigma(t)^2 = \cos \left( \pi t \right) - \sin \left( \pi t \right)\mathbf{k} \]

The endpoint of this lift is $\sigma(1)^2 = -1$. Similarly the endpoint of the lift of a 540 shove-it is $\sigma(1)^3 = \mathbf{k}$. We will show that because these tricks exhaust the possibilities of the endpoints of quaternionic lifts, they will (together with the ollie) exhaust the homotopy classes of quaternionic flips. As a result their associated topological flips will in turn exhaust the homotopy classes of topological flips, which is the content of the main theorem.

\begin{figure}[ht]
    \vspace{-6pt}
	\centering
	\begin{subfigure}[t]{0.31\textwidth}
		\centering
		\includegraphics[width=\textwidth]{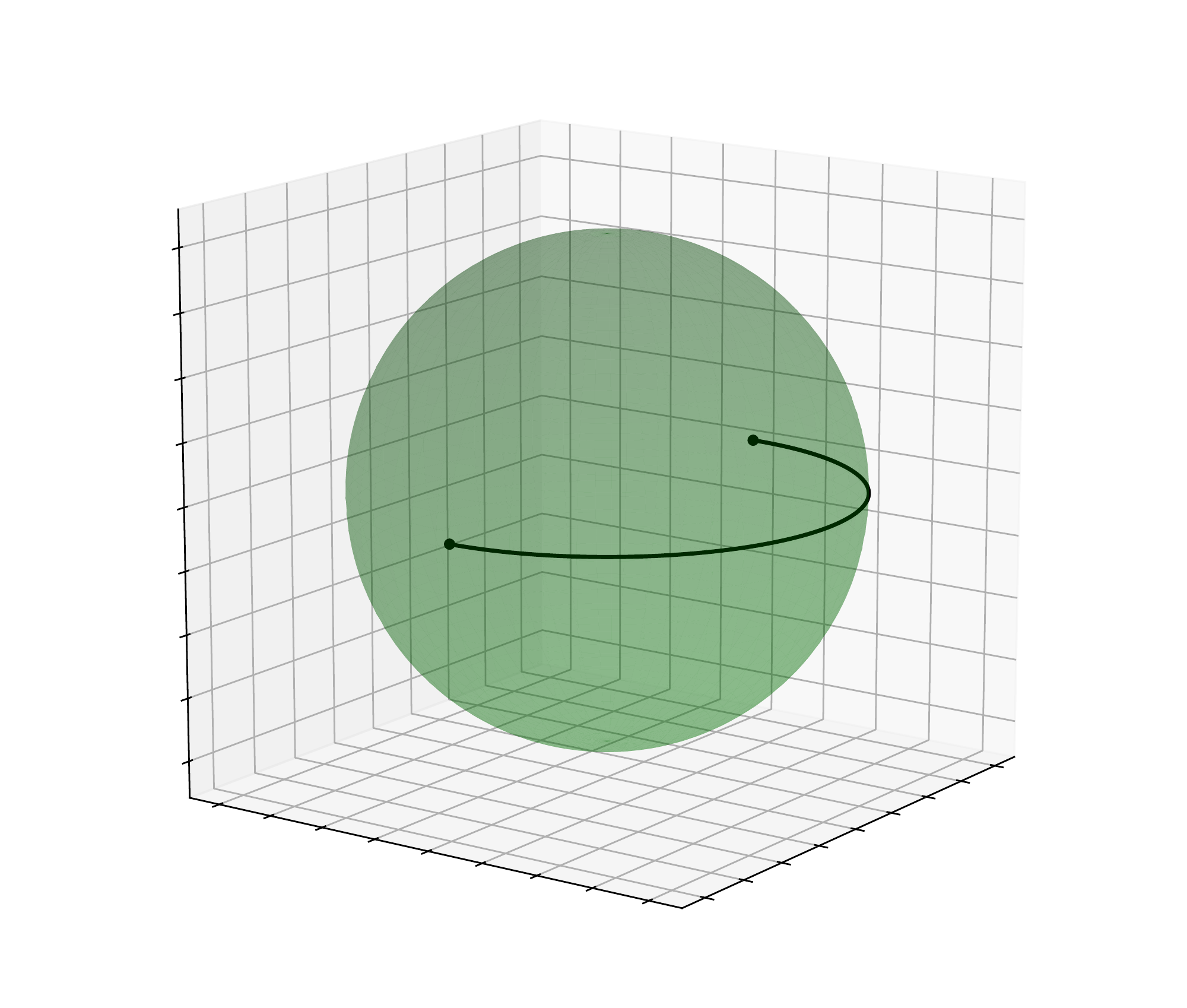}
		\caption{$\kappa$ curve}		
	\end{subfigure}
	\quad
	\begin{subfigure}[t]{0.31\textwidth}
		\centering
		\includegraphics[width=\textwidth]{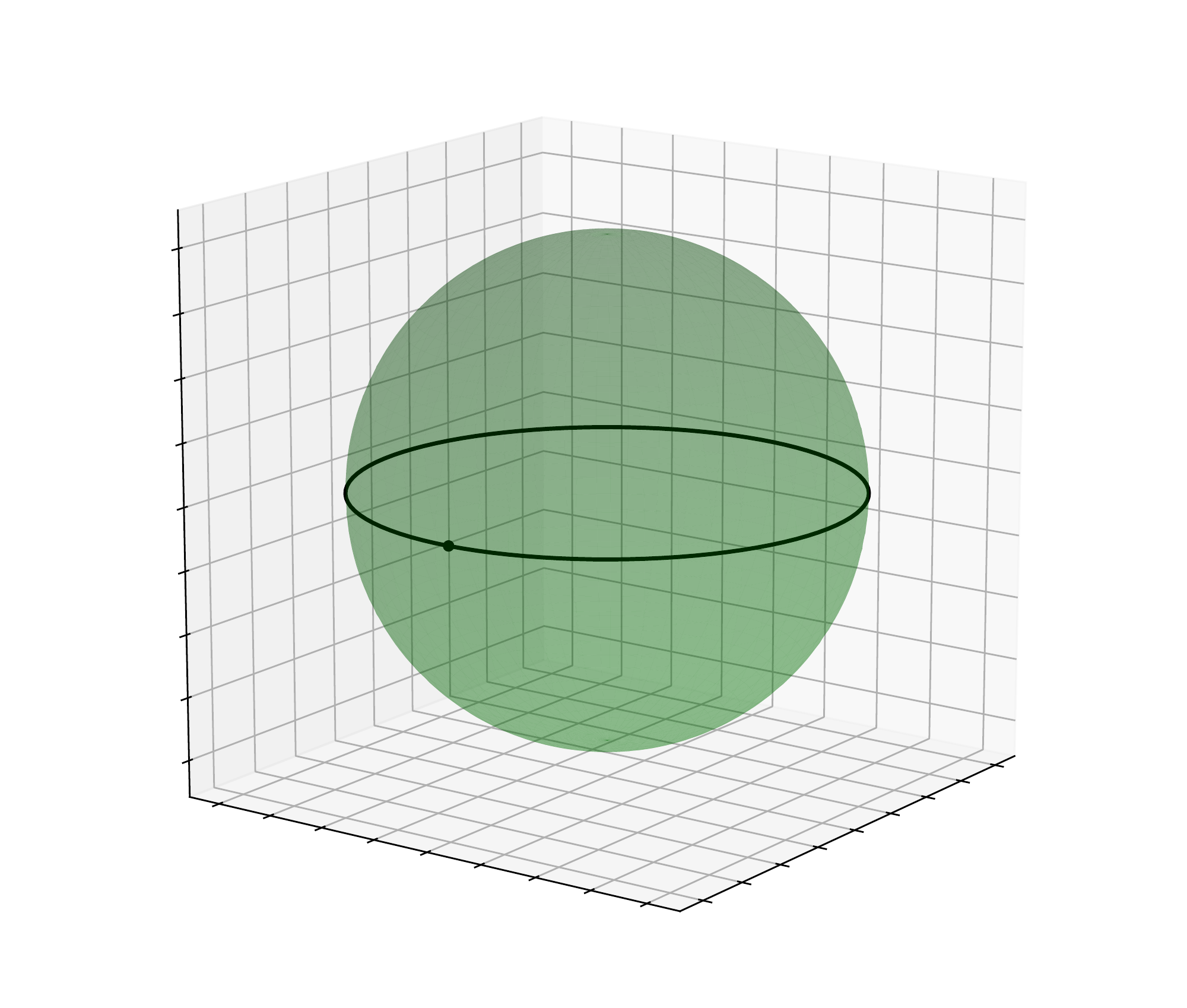}
		\caption{$\kappa^2$ curve}
	\end{subfigure}
	\quad
	\begin{subfigure}[t]{0.31\textwidth}
		\centering
		\includegraphics[width=\textwidth]{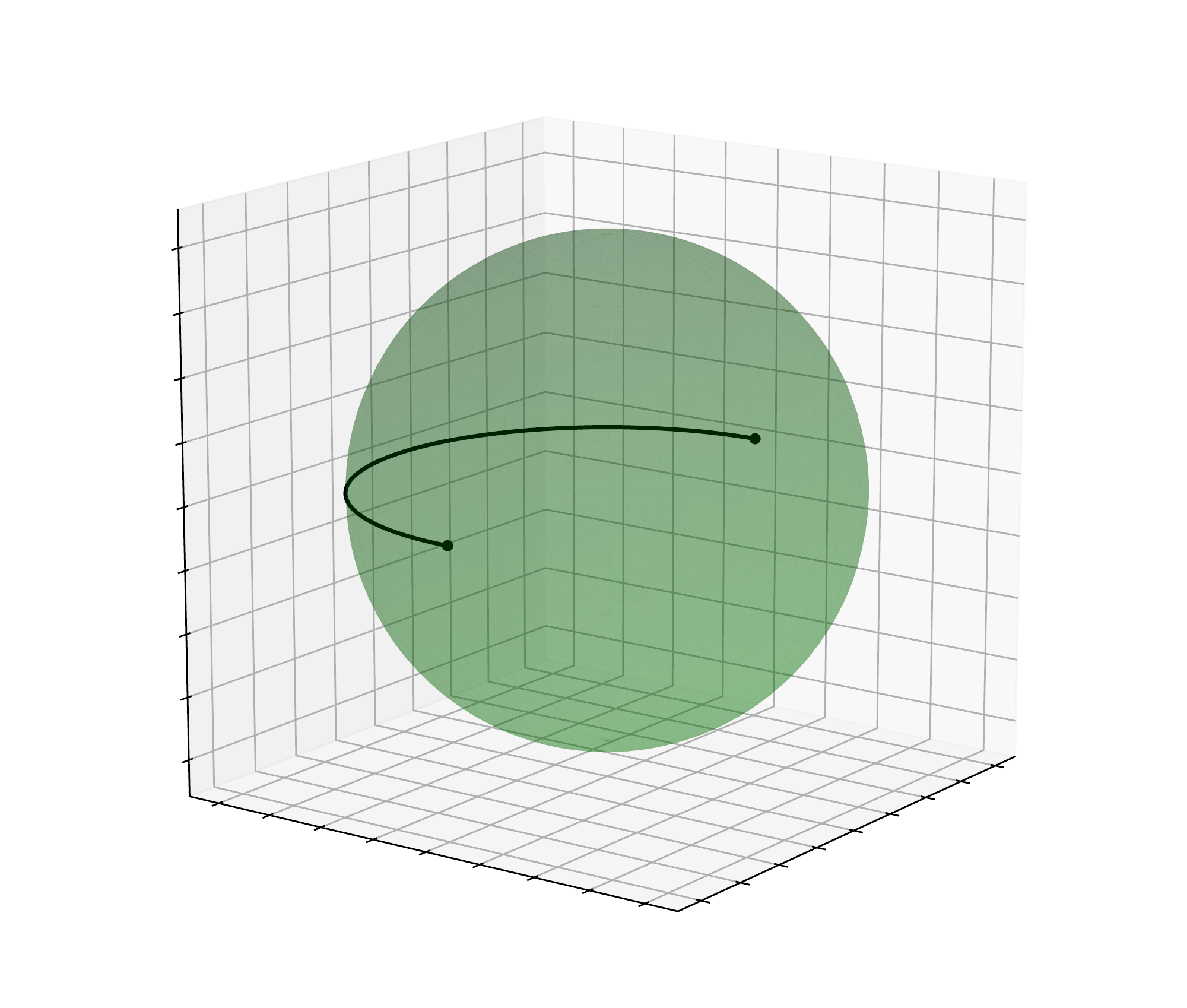}
		\caption{$\kappa^{-1}$ curve}
	\end{subfigure}
	\caption{Quaternionic lifts of the powers of a kickflip}
	\label{fig:kickflips}
	\vspace{-4pt}
\end{figure}

Recall that the kickflip $K(t)$ consists of a left-handed 360 degree rotation of the skateboard about the $y$-axis. Its quaternionic lift is:
    \[ \kappa(t) = \cos (\pi t) - \sin (\pi t)\mathbf{j}. \]

The lift of the double-kickflip can be described by $\kappa(t)^2 = \cos (2\pi t) - \sin (2\pi t)\mathbf{j}$. The heelflip, which is the kickflip with a reversed rotation, that is, a right-handed 360 degree rotation about the $y$-axis, has the quaternionic lift $\kappa(t)^{-1} = \cos (\pi t) + \sin (\pi t)\mathbf{j}$. 
    
These three curves all have vanishing $\mathbf{i}$ component, so we are able to visualize them by using the same 3-dimensional subspace as before, see figure \ref{fig:kickflips}.
    
The varial kickflip was modeled by $S(t)K(t)$. Because $\rho$ is a group homomorphism, we deduce that its quaternionic lift is given by $\sigma(t)\kappa(t)$. This curve does not remain in the 3-dimensional subspace generated by $1$, $\mathbf{j}$ and $\mathbf{k}$, however the orthogonal projection onto this subspace is injective when restricted to this curve, so that we are still able to obtain a good visualization of this curve. We will also renormalize the curve's projection so that it remains on the surface of the unit 2-sphere, see figure \ref{fig:varial}.

\begin{figure}[ht]
    \vspace{-2pt}
	\centering
	\includegraphics[width=.5\textwidth]{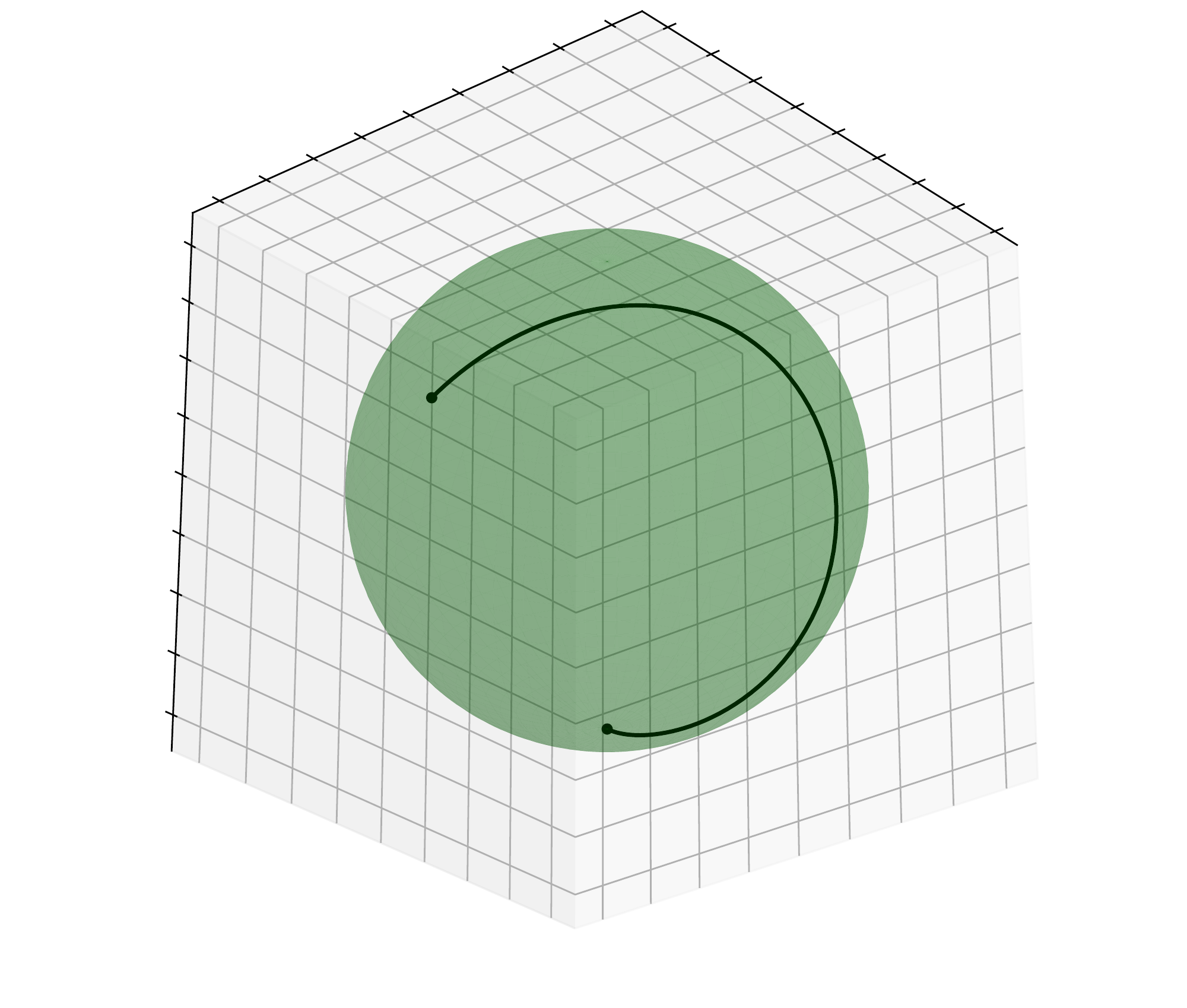}
	\vspace{-16pt}
	\caption{Renormalized projection of the quaternionic varial kickflip}
	\label{fig:varial}
	\vspace{-6pt}
\end{figure}

Lastly we consider the hardflip. The right-handed 180 degree rotation about the $x$-axis $U(t)$ has the quaternionic lift:
    \[ \upsilon(t) = \cos \left( \frac{\pi t}{2} \right) + \sin \left( \frac{\pi t}{2} \right) \mathbf{i} \]

We then obtain that the quaternionic lift of the hardflip is $\upsilon(t)\kappa(t/2)$.

\section{Proof of the Four Tricks Theorem}
\label{sec:fourtricks}

In order to prove the four tricks theorem it will suffice to show that $[QF]$ has cardinality equal to four. In fact we will show that the following map:
\[\begin{array}{crcl}
 \text{ev}_1:    & [QF]    & \to       & \langle \textbf{k} \rangle  \\
                & \gamma        & \mapsto   & \gamma(1)
\end{array}\]

\noindent is a bijection. Furthermore we will later describe a geometrically defined group structure on $[TF]$ which will make the composition $\text{ev}_1 \circ \lambda$ into a group homomorphism.

In order to do this, first recall that given a topological space $X$, and two continuous curves $f$ and $g$ defined on the interval $[0, 1]$ such that $f(1) = g(0)$, we may define another continuous curve called their \emph{concatenation} by:
\[(f\# g)(t) = \begin{cases}
                f(2t) \text{ for } t \in [0,1/2]    \\
                g(2t-1) \text{ for } t \in [1/2,1]
                \end{cases}\]

Furthermore given a continuous curve $f$, we denote by $\bar{f}$ the curve which traverses $f$ in the opposite direction, that is $\bar{f}(t) = f(1-t)$.

Recall that a loop $f$ in $X$ based at a point $p \in X$ is a continuous curve defined on the interval $[0, 1]$ such that $f(0)=f(1)=p$. A loop based at $p$ is contractible if it is homotopic to the constant loop at $p$.

\begin{proposition}
Let $X$ be a topological space, $p$ and $q$ be two points in $X$, and let $f$ and $g$ be two continuous curves defined on $[0, 1]$ from $p$ to $q$. Then $f$ and $g$ are homotopic if and only if the loop $f \# \bar{g}$ is contractible. 
\end{proposition}

\begin{proof}
Suppose that $f$ and $g$ are homotopic, then $f \# \bar{g}$ is homotopic to $g \# \bar{g}$ which is contractible via the homotopy:
    \[ H(s,t) = \begin{cases}
                g(2st) \text{ for } t \in [0,1/2] \\
                g(2s(1-t)) \text{ for } t \in [1/2,1]
                \end{cases}\]

Since $H(0,t) = p$ for all $t \in [0, 1]$ and $H(1,t) = g \# \bar{g}(t)$.\\

Conversely, suppose that $f \# \bar{g}$ is contractible. Since $\bar{g} \# g$ is contractible we know that:
\[[f] = [ f \# (\bar{g} \# g) ] = [ (f \# \bar{g}) \# g ] = [g]\]

\end{proof}

We may now prove the four tricks theorem by using the fact that in the 3-sphere every loop is contractible (that it $S^3$ is simply-connected).

\begin{proof}[Proof of the four tricks theorem]

We show that the map $\text{ev}_1$ is bijective. We start by showing that it is injective. Let $f$ and $g$ be two quaternionic flips such that $\text{ev}_1([f])=\text{ev}_1([g])$. Since $f(1)=g(1)$ we may consider the loop $f \# \bar{g}$. Because the 3-sphere  is simply-connected this loop is contractible, so that $[f]=[g]$.

In order to show that $\text{ev}_1$ is surjective, recall that the quaternionic lift of the shove-it $S(t)$ is given by $\sigma(t) = \cos(\pi t/2)-\sin(\pi t/2)\textbf{k}$, so that $\text{ev}_1([\sigma]) = -\textbf{k}$. Furthermore, given any integer $n \in \mathbb{Z}$ we may consider the quaternionic flip $\sigma(t)^n$ and $\text{ev}_1([\sigma^n]) = (-\textbf{k})^n$, so that $\text{ev}_1$ is surjective.

\end{proof}

\section{The Group Structure}
\label{sec:groupstructure}

We may define a group structure in the set $[TF]$ of homotopy classes of topological flips in much the same way as one does with the fundamental group of a topological space. The difference here is that if a topological flip $f$ lands at the configuration $\mathcal{O}$, that is if $f(1)=\mathcal{O}$, we are not able to concatenate $f$ with another topological flip $g$, since $g$ must start at the configuration $I_3$. We may however concatenate the curve $f$ with the trick $g_\mathcal{O}(t) = g(t)\cdot \mathcal{O}$ which is a curve that models a trick which looks almost identical to $g$, but at the start of the trick the skater has their back foot resting on the \textit{nose} of the skateboard instead of its \textit{tail} (the tail would then point towards the positive $y$-axis). If the skate is symmetric in the sense that its front side is exactly equal to its back side, then the trick $g_\mathcal{O}$ would be visually indistinguishable from $g$.

\begin{definition}
Let $[f]$ and $[g]$ be homotopy classes of two topological flips, we define their multiplication by:
\[ [f]\ast[g] = \begin{cases}
                [f \# g] & \text{ \emph{if} } f(1) = I_3 \\
                [f \# g_\mathcal{O}] & \text{ \emph{if} } f(1) = \mathcal{O}
\end{cases}\]
\end{definition}

\textbf{Remark:} If we allowed our flip tricks to start at the position $\mathcal{O}$ (with the back foot resting on the nose of the skateboard) we could use the coset space $SO(3)/\langle \mathcal{O} \rangle$ as our configuration space instead of $SO(3)$. We have not used that space in this text for two reasons. First: its elements would no longer be represented by matrices but cosets which is a bit inconvenient. Second: it does not have a natural group structure (since $\langle \mathcal{O} \rangle$ is not a normal subgroup of $SO(3)$), which would make defining more complex topological tricks more cumbersome. There is an advantage of working in the space $SO(3)/\langle \mathcal{O} \rangle$ however, in that case the topological flips can be identified with loops in $SO(3)/\langle \mathcal{O} \rangle$ and the group structure on the homotopy classes of topological flips is perhaps more ``standard'', as it would agree with the operation of the fundamental group $\pi_1(SO(3)/\langle \mathcal{O} \rangle)$.

We will omit the proof that this group operation is well defined (it does not depend on the choice of representative of the homotopy classes and satisfies the axioms of a group), as it is very similar to the argument used for the operation of the fundamental group of a topological space (see \cite{Ha} Proposition 1.3).

Recall that $\lambda$ is the ``lifting bijection'' $\lambda: [TF] \to [QF]$. We now define the map $\Lambda = \text{ev}_1 \circ \lambda$:
\[\begin{array}{crcl}
 \Lambda:    & [TF]    & \to       & \langle \textbf{k} \rangle
\end{array}\]

We will now show that this map is compatible with the group structure that we have defined on $[TF]$.

\begin{proposition}
The map $\Lambda$ is a group isomorphism.
\end{proposition}

\begin{proof} We have already shown that both $\lambda$ and $\text{ev}_1$ are bijective, so that $\Lambda$ is bijective. We must then show that $\Lambda$ is a group homomorphism.

We have seen that any quaternionic flip is homotopic to a flip of the form $\sigma(t)^n$ for some integer $n \in \mathbb{Z}$. Because $\lambda$ is bijective, we conclude that any topological flip is homotopic to a flip of the form $S(t)^n$. Since there are only four homotopy classes in total, it suffices to prove that $[S(t)]^n = [S(t)^n]$ for integers $ 0 \leq n \leq 3$. This would imply that $\Lambda([S(t)]^n) = \Lambda([S(t)^n]) = \text{ev}_1(\sigma(t)^n) = (-\textbf{k})^n$, so that $\Lambda$ will in fact be a group isomorphism.

The cases $n=0$ and $n=1$ are trivial. We then start by studying the case $n=2$. We see that $[S]^2 = [S \# S_\mathcal{O}]$, as $S(1)= \mathcal{O}$. Notice that for $t \in [0,1/2]$ we have:
\[S \# S_\mathcal{O} (t) = S(2t)
= \begin{bmatrix}
\cos(2\pi t)     & \sin(2\pi t)   & 0 \\
-\sin(2\pi t)    & \cos(2\pi t)   & 0 \\
0               & 0             & 1
\end{bmatrix}\]

And for $t \in [1/2,1]$:
\[\begin{aligned}
    S \# S_\mathcal{O} (t) & = S_\mathcal{O} (2t-1) = S(2t-1)\cdot \mathcal{O}\\
    & = \begin{bmatrix}
    \cos(2\pi t-\pi)     & \sin(2\pi t -\pi)   & 0 \\
    -\sin(2\pi t -\pi)    & \cos(2\pi t -\pi)   & 0 \\
    0               & 0             & 1
    \end{bmatrix} \cdot \begin{bmatrix}
    -1	&	0	&	0	\\
    0	&	-1	&	0	\\
    0	&	0	&	1
    \end{bmatrix} \\
    & = \begin{bmatrix}
    -\cos(2\pi t)     & -\sin(2\pi t)   & 0 \\
    \sin(2\pi t)    & -\cos(2\pi t)   & 0 \\
    0               & 0             & 1
    \end{bmatrix}\cdot \begin{bmatrix}
    -1	&	0	&	0	\\
    0	&	-1	&	0	\\
    0	&	0	&	1
    \end{bmatrix}\\
    & = \begin{bmatrix}
    \cos(2\pi t)     & \sin(2\pi t)   & 0 \\
    -\sin(2\pi t)    & \cos(2\pi t)   & 0 \\
    0               & 0             & 1
    \end{bmatrix}
    \end{aligned}\]
    
We conclude that for all $t \in [0,1]$ we have $S \# S_\mathcal{O} (t) = S(t)^2$, so that in particular at the level of homotopy classes $[S]^2 = [S \# S_\mathcal{O}] = [S^2]$.

We now show that $[S(t)]^3 = [S(t)^3]$. We notice that $[S(t)]^3 = [S(t)]^2 \ast [S(t)] = [S^2 \# S(t)]$. We may ``spread out'' both the 360 shove-it $S^2$ and the shove-it $S$, so that they take place simultaneously along the entire interval $[0,1]$, instead of only in one half of the interval. We can then deform the curve $S^2 \# S(t)$ into the curve $S^3(t)$.

Consider the homotopy:
\[ F(s,t) = \begin{cases}
            S^2\big( 2t/(s+1) \big) & \text{ for } 0 \leq t \leq (s+1)/2 \\
            I_3 & \text{ otherwise}
            \end{cases}\]

When $s = 0$ the curve $F(0,t)$ performs a 360 shove-it during the interval $[0,1/2]$ and is constant at $I_3$ on the interval $[1/2,1]$, while when $s=1$ the curve is simply equal to a 360 shove-it, $F(1,t) = S^2(t)$.

We may do something similar for the shove-it
\[ G(s,t) = \begin{cases}
            S\big( (2t+s-1)/(s+1) \big) & \text{ for } (1-s)/2 \leq t \leq 1 \\
            I_3 & \text{ otherwise}
            \end{cases}\]

We then consider the homotopy $H(s,t) = F(s,t) \cdot G(s,t)$. We see that $H(0,t) = S^2 \# S(t)$ and $H(1,t) = S(t)^3$, so that $[S]^3 = [S^2 \# S] = [S^3]$ which finishes the proof.

\end{proof}

We then obtain the simple corollary:

\begin{corollary}
The group $[TF]$ is isomorphic to the cyclic group of order four $\mathbb{Z}/4\mathbb{Z}$.
\end{corollary}

\section{Examples of homotopies}
\label{sec:homotopies}

In this section we describe a few explicit homotopies between a number of different tricks. Recall that the double kickflip is a trick where the skateboard performs two full left-handed rotations about the y-axis:
\[K(t)^2 = \begin{bmatrix}
\cos(4\pi t)    & 0 & -\sin(4\pi t) \\
0               & 1 & 0             \\
\sin(4\pi t)    & 0 & \cos(4\pi t)
\end{bmatrix}\]

Its quaternionic lift is:
\[\kappa^2(t) = \cos(2\pi t) - \sin(2\pi t)\textbf{j}\]

Notice that since $\kappa^2(1)=1$ we know that $\kappa^2$ is homotopic to the constant loop based at $1$. We may construct this homotopy geometrically by noticing that the curve $\kappa^2(t)$ describes a great circle in the 3-sphere. We may deform it to the constant loop $O(t)=I_3$ by contracting this  great circle through the upper hemisphere of $S^3$. We then obtain the following homotopy $H_{O \to K^2}: [0,1]\times[0,1] \to SO(3)$:
\[H_{O \to K^2}(s,t) = \rho\big( 1 + s(\cos(2\pi t)-1), 0, -s\sin(2\pi t), -\sqrt{2s(1-s)(1-\cos(2\pi t))}\, \big)\]

This homotopy remains in the 3-subspace generated by $1, \mathbf{j}$ and $\mathbf{k}$ so we may visualize it as before. See figure \ref{fig:2kickcontraction}.
\begin{figure}[ht]
    \vspace{-12pt}
	\centering
	\includegraphics[width=.5\textwidth]{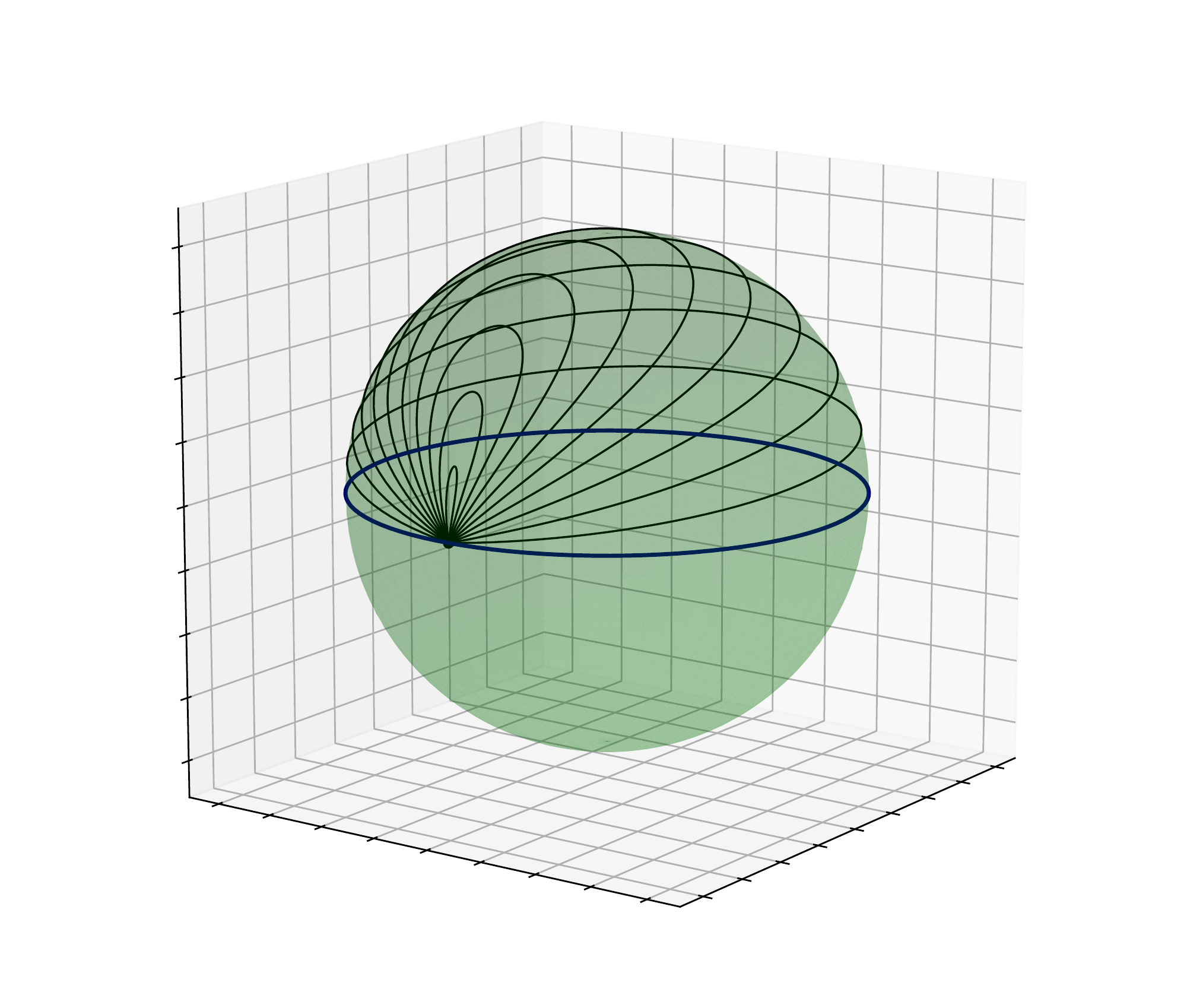}
	\vspace{-20pt}
	\caption{Contraction of the double kickflip}
	\label{fig:2kickcontraction}
	\vspace{-6pt}
\end{figure}

Because the double kickflip $K^2(t)$ is homotopic to the constant loop $O(t)$ we know that $2[K] = [K^2] = [O] = 0$, and so $[K]=-[K]$. This means that there is a deformation between the kickflip $K(t)$ and the reversed trick, the heelflip $K(1-t) = K(t)^{-1}$. We have seen that their quaternionic lifts form great semicircles in the plane generated by $1$ and $\mathbf{j}$. We can deform one semicircle into the other by rotating them about the $\textbf{k}$ axis, see figure \ref{fig:kicktoheel}. We obtain the homotopy $H_{K \to K^{-1}}: [0,1]\times[0,1] \to SO(3)$:
\[H_{K \to K^{-1}}(s,t) = \rho\big( \cos(\pi t), 0, -\sin(\pi t)\cos(\pi s), -\sin(\pi t)\sin(\pi s) \big)\]

\begin{figure}[ht]
    \vspace{-12pt}
	\centering
	\includegraphics[width=.5\textwidth]{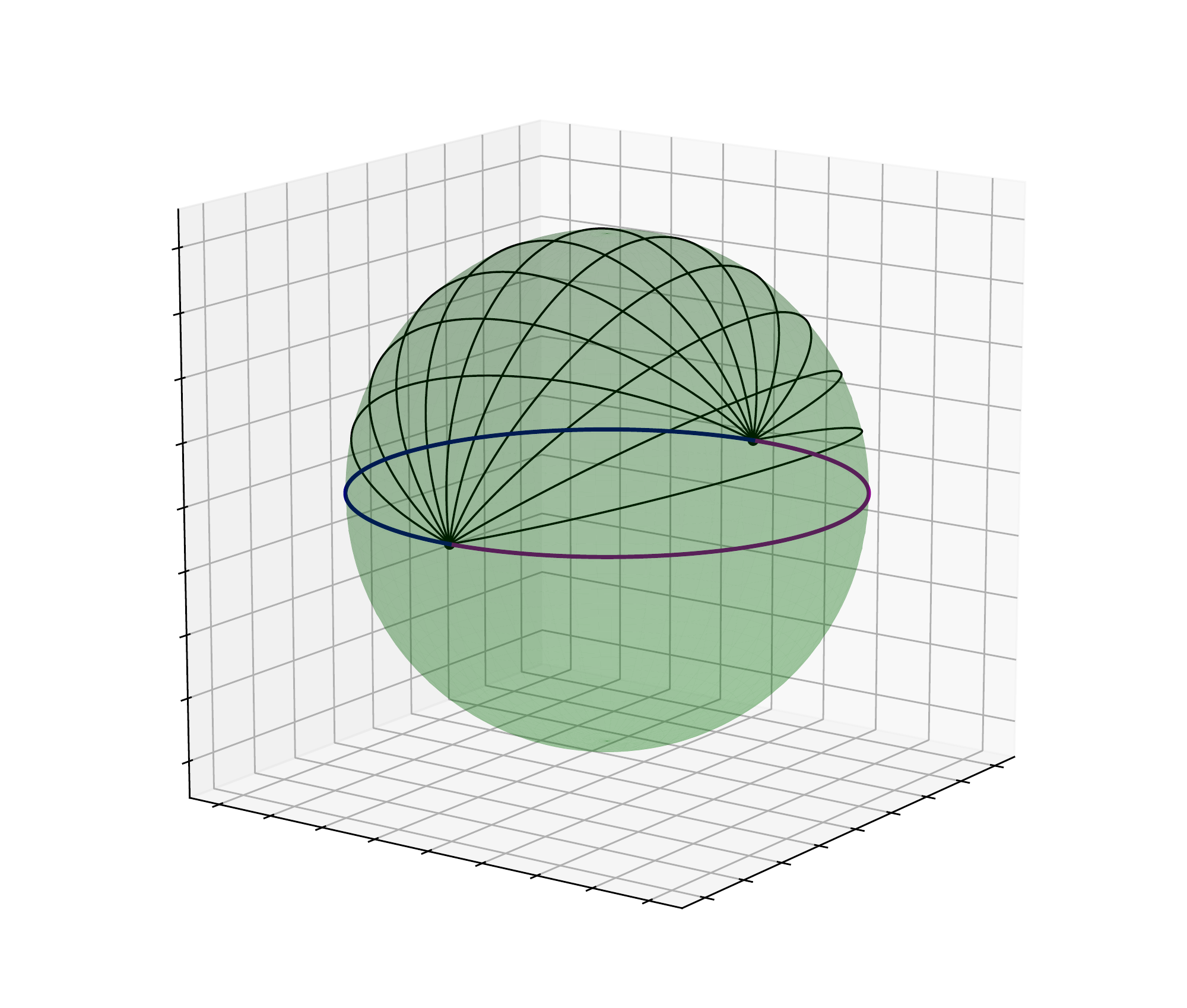}
	\vspace{-20pt}
	\caption{Deformation from a kickflip to a heelflip}	
	\label{fig:kicktoheel}
	\vspace{-6pt}
\end{figure}

Notice that in the deformation above as the parameter $s$ moves from $0$ to $1/2$ we may find a homotopy between the kickflip and the 360 shove-it. The quaternionic lift of both of these tricks form a great circle in the 3-sphere, as the parameter $s$ increases, the axis of rotation $-\mathbf{j}$ of the kickflip continuously moves towards the axis of rotation $-\mathbf{k}$ of the 360 shove-it. By rescaling the parameter $s$ so that it lies between $0$ and $1$, we obtain the following homotopy:
\[H_{K \to S^2}(s,t) = \rho\big( \cos(\pi t), 0, -\sin(\pi t)\cos(\pi s/2), -\sin(\pi t)\sin(\pi s/2) \big)\]

As a last example we can use the multiplication in $SO(3)$ to deform the varial kickflip $V(t) = S(t)K(t)$ into a 540 shove-it $S(t)^3$ by simply multiplying the previous homotopy by $S(t)$. We obtain $H_{V \to S^3}(s,t) = S(t)H_{K \to S^2}(s,t)$, see figure \ref{fig:varialto540shuv}.

\begin{figure}[ht]
    \vspace{-12pt}
	\centering
	\includegraphics[width=.5\textwidth]{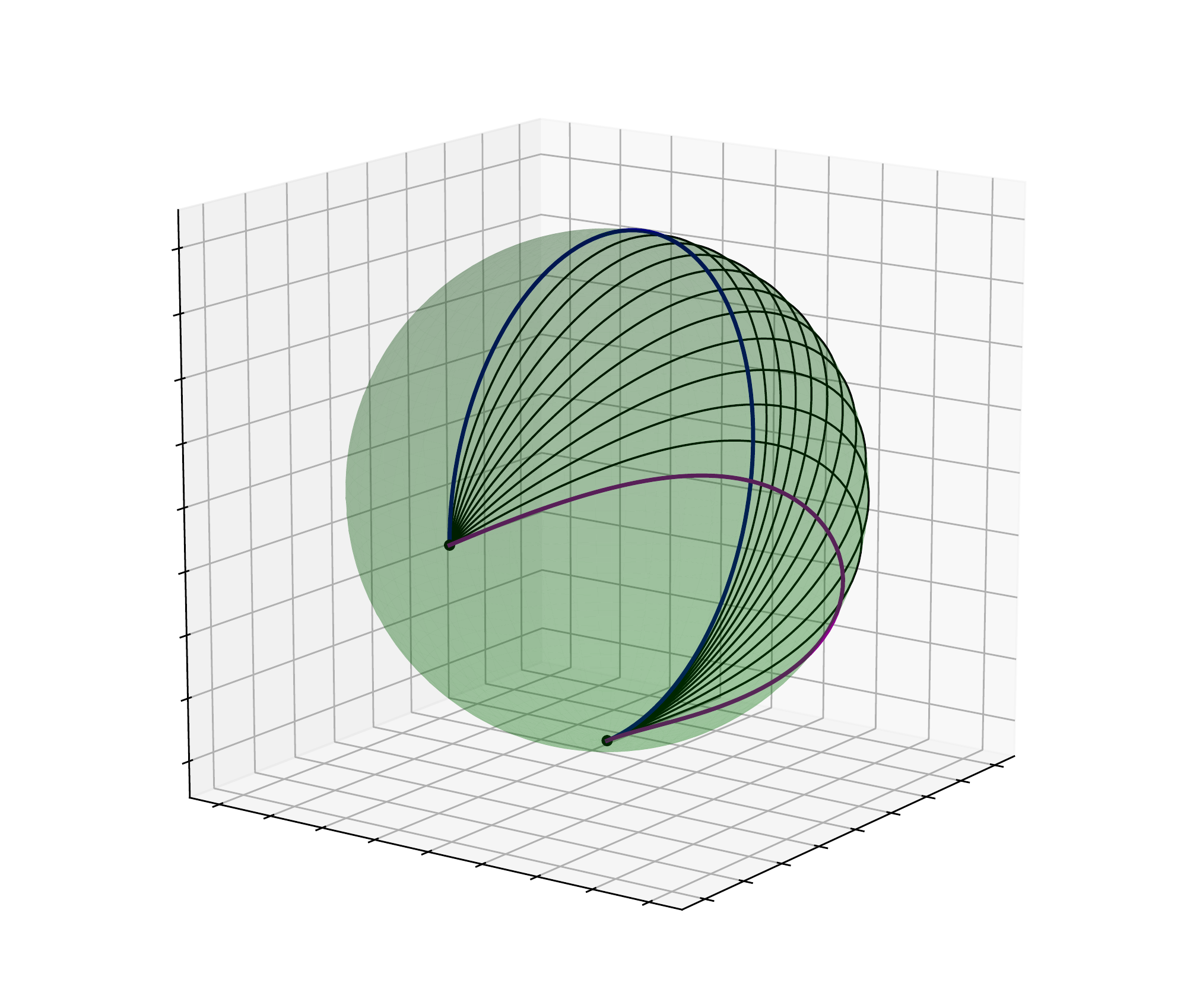}
	\vspace{-20pt}
	\caption{Deformation from a varial kickflip to a 540 shove-it}
	\label{fig:varialto540shuv}
	\vspace{-10pt}
\end{figure}
%

\section{Stabilizing Rotations}
\label{sec:stabilizing}

We can see that any topological flip can be deformed into a rotation about the $\textbf{k}$ axis with constant angular velocity. In this section we construct one more interesting homotopy by exploring the geometry of the 3-sphere. This will be a deformation between a somewhat generic topological flip and one which rotates with constant angular velocity about a fixed axis of rotation. We then study an example consisting of a shove-it with a wobbling axis of rotation and make use of this homotopy to stabilize its axis and deform it into the standard shove-it.

Consider a unit vector $\mathbf{u}$ in $\mathbb{R}^3$ and let $P^\perp_\mathbf{u}$ be the 2-plane perpendicular to $\mathbf{u}$. We may also identify $P^\perp_\mathbf{u}$ with the corresponding subset of the imaginary quaternions, we then consider the circle $C^\perp_\mathbf{u} = \rho(S^3 \cap P^\perp_\mathbf{u})$ inside of $SO(3)$. We also define the circle $C^\parallel_\mathbf{u} = \rho(S^3 \cap P^\parallel_\mathbf{u})$ where $P^\parallel_\mathbf{u}$ is the 2-plane in $\mathbb{H}$ generated by $1$ and $\mathbf{u}$.

Finally we let $\mathbf{a}, \mathbf{b}$ be an orthonormal basis of $P^\perp_\mathbf{u}$ and consider the orthonormal basis $\beta_\mathbf{u} = (\mathbf{a}, \mathbf{b}, \mathbf{u})$ of $\mathbb{R}^3$.

\begin{proposition} Let $A: [0,1] \to SO(3)$ be a continuous curve such that $A(0)$ and $A(1)$ belong to $C^\parallel_\mathbf{u}$, and suppose that $A$ does not intersect the circle $C^\perp_\mathbf{u}$. There are numbers $a \in \mathbb{R}$ and $\varphi \in [0,2\pi)$ such that $A(t)$ is homotopic to the curve $S_{a,\varphi}(t)$ which in the basis $\beta_\mathbf{u}$ is given by:
\[ \Big[S_{a,\varphi}(t)\Big]_{\beta_\mathbf{u}} = \begin{bmatrix}
\cos(at+\varphi)	&	-\sin(at+\varphi)	&	0	\\
\sin(at+\varphi)	&	\cos(at+\varphi)	&	0	\\
0			        &	0			        &	1
\end{bmatrix} \]
\end{proposition}

\begin{proof}

We will use the following coordinates for the quaternions:
\[ q = x_0 + x_1\mathbf{a} + x_2\mathbf{b} + x_3\mathbf{u} \]

Let $\Gamma^\perp_\mathbf{u} = S^3 \cap P^\perp_\mathbf{u}$ and $\Gamma^\parallel_\mathbf{u} = S^3 \cap P^\parallel_\mathbf{u}$, so that $\Gamma^\perp_\mathbf{u} = \rho^{-1}(C^\perp_\mathbf{u})$ and $\Gamma^\parallel_\mathbf{u} = \rho^{-1}(C^\parallel_\mathbf{u})$. Using our coordinates we may write:
\[ \Gamma^\perp_\mathbf{u} = \left\{ q \in S^3 \mid x_{0}^{2} + x_{3}^{2} = 0 \right\}. \]

Since over the open set $D_\mathbf{u} = S^3 - \Gamma^\perp_\mathbf{u}$ of $S^3$ we have the inequality $x_{0}^{2} + x_{3}^{2} \neq 0$, we are able to define a homotopy $F: [0,1] \times D_\mathbf{u} \to D_\mathbf{u}$ given by
\[ F(s, q) = \left( \sqrt{n(s,q)} \cdot x_{0}, \sqrt{1-s} \cdot x_{1}, \sqrt{1-s} \cdot x_{2}, \sqrt{n(s,q)} \cdot x_{3} \right) \]

\noindent where the function $n: [0,1] \times D_\mathbf{u} \to \mathbb{R}_{\geq 0}$ is given by:
\[ n(s, q) = \frac{1 - (1-s)(x_{1}^{2} + x_{2}^{2})}{x_{0}^{2} + x_{3}^{2}}. \]

Denote $F_s(q) = F(s,q)$. Notice that when $s=0$ the function $F_0$ is the identity map on $D_\mathbf{u}$, and when $s=1$ the function $F_1$ maps $D_\mathbf{u}$ onto the circle $\Gamma^\parallel_\mathbf{u}$. Moreover, for any value of the parameter $s$, the map $F_s$ leaves the points on the circle $\Gamma^\parallel_\mathbf{u}$ fixed. The function $F$ is called a \emph{deformation retraction} of $D_\mathbf{u}$ onto $\Gamma^\parallel_\mathbf{u}$.

Now let $\alpha: [0,1] \to S^3$ be a quaternionic lift of $A$. We may define a homotopy $\hat{F}: [0, 1] \times [0,1] \to SO(3)$ by:
\[ \hat{F}(s,t) = \rho \circ F(s, \alpha(t)).\]

This is a homotopy from $A(t)$ to a curve entirely contained in $C^\parallel_\mathbf{u}$.\\

At this stage, we know that $\alpha(t)$ is homotopic to a curve of the form 
\[ \alpha_1(t) = \cos( \theta (t)) + \sin( \theta (t)) \textbf{u} \]

\noindent where $\theta (t) : [0,1] \to \mathbb{R}$ is a continuous function. All that is left for us to do is to deform this curve into one which has constant angular velocity, that is, a curve where the rotation angle is linear in the time $t$. We may accomplish this by using a convex combination:
\[ \gamma(s,t) = (1-s)\theta(t) + s(a_0t + \varphi_0) \]

\noindent where $a_0 = \theta(1) - \theta(0)$ and $\varphi_0 = \theta(0)$. We proceed by defining the following homotopy $G: [0, 1] \times [0,1] \to SO(3)$:
\[ G(s,t) = \rho \big( \cos ( \gamma(s,t) ) + \sin ( \gamma(s,t) ) \textbf{u} \big) \]

We see that $G(1, t)$ is equal to the curve 
\[ \cos (a_0 t + \varphi_0) + \sin (a_0 t + \varphi_0) \textbf{u} \]

By concatenating the homotopies $\hat{F}$ and $G$, we find the homotopy claimed in the proposition, where by the Rodriguez formula we must have $a = a_0/2$ and $\varphi = \varphi_0/2$.
\end{proof}

We now discuss the example of a 360 shove-it with a wobbling axis of rotation. The 360 shove-it rotates about the axis $-\mathbf{k}$ by an angle of 360 degrees with a right-handed orientation. We now consider a trick which is close to the 360 shove-it, but whose axis of rotation wobbles about the $z$-axis. We may consider for example the following formula for the axis of rotation:
\[ \mathbf{k}_{a,\omega}(t) = a\cos(2\pi\omega t)\textbf{i} + a\sin(2\pi\omega t)\textbf{j} - \sqrt{1-a^{2}}\textbf{k} \]

\noindent where $\omega$ is a real number which determines the frequency of the wobble and $0 \leq a \leq 1$ determines its amplitude. When $a=0$ there is no wobble, and when $a=1$ the resulting trick no longer bears any resemblance to a 360 shove-it. The resulting trick would then be:
\[ R(t) = \rho\big( \cos(\pi t) + \sin(\pi t)\mathbf{k}_{a,\omega}(t)\big) \]

We may project the quaternionic curve onto the 3-subspace generated by $1$, $\mathbf{j}$ and $\mathbf{k}$ as before, and obtain the visualization in figure \ref{fig:wobblingshuv}.
\begin{figure}[ht]
    \vspace{-12pt}
	\centering
	\includegraphics[width=.5\textwidth]{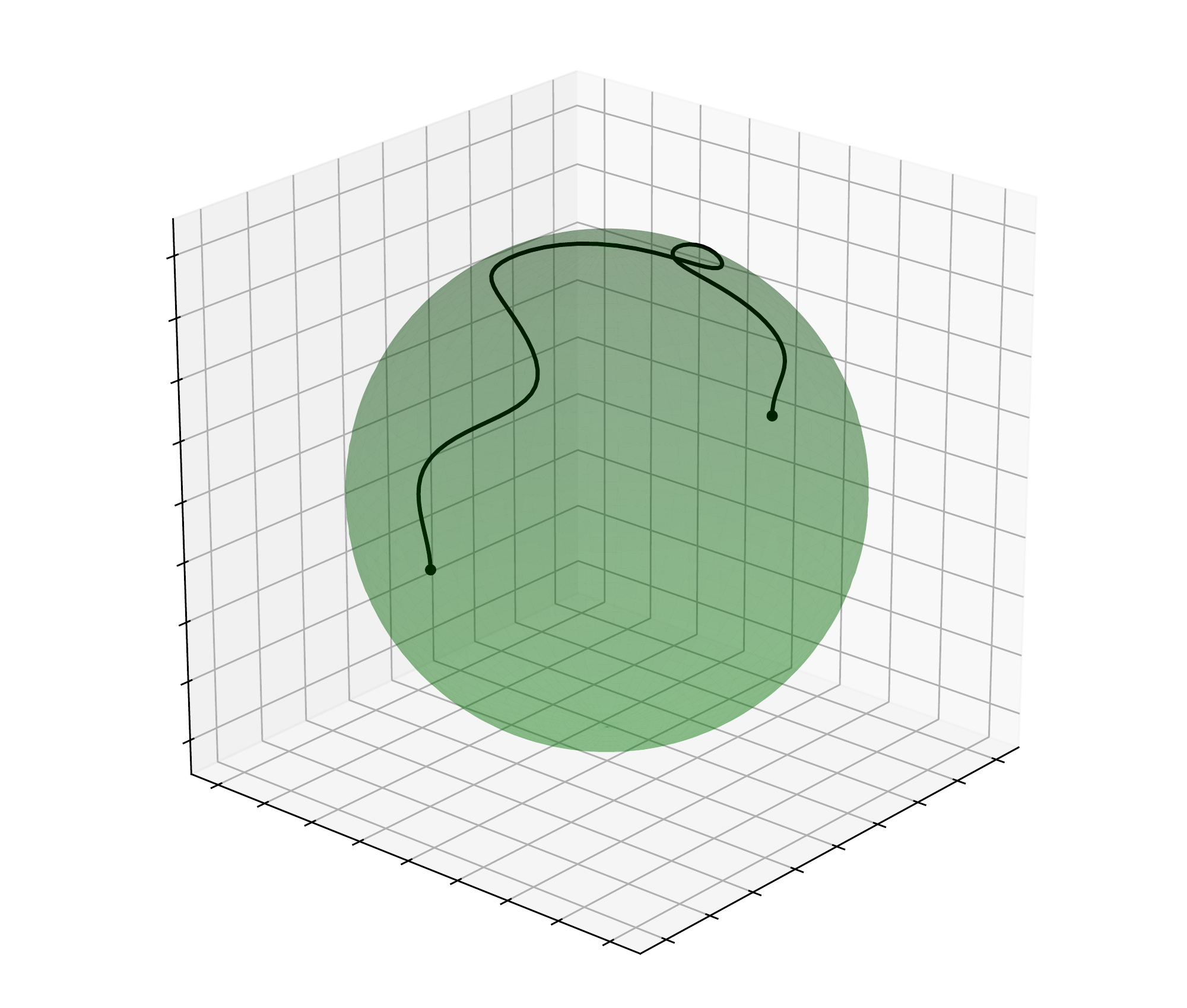}
	\vspace{-12pt}
	\caption{Renormalized projection of wobbling 360 shove-it}
	\label{fig:wobblingshuv}
	\vspace{-6pt}
\end{figure}

Let's now stabilize the axis of rotation of this trick. We will simply apply the homotopy $\hat{F}$ from the proposition in order to create a visualization of that deformation. In this case, the axis $\mathbf{u}$ in the proposition is $-\mathbf{k}$ and we may use the standard coordinates for the quaternions in our computation. The quaternionic curve is $\alpha(t) = \cos(\pi t) + \sin(\pi t)\mathbf{k}_{\omega,a}(t)$. We first compute the function $n(s,t) := n(s,\alpha(t))$, we obtain:
\[ n(s,t) = 1 + \frac{s}{(1/a\sin(\pi t))^2 - 1} \]

The homotopy is then given by:
\[ H(s,t) = \left( \sqrt{n(s,t)} \cdot \alpha_0(t), \sqrt{1-s} \cdot \alpha_1(t),
                \sqrt{1-s} \cdot \alpha_2(t), \sqrt{n(s,t)} \cdot \alpha_3(t) \right) \]

We will not reparametrize the angle in order to obtain a curve with constant angular velocity as this will not affect the image produced. We obtain figure \ref{fig:wobblehomotopy}.
\begin{figure}[ht]
    \vspace{-8pt}
	\centering
	\includegraphics[width=.5\textwidth]{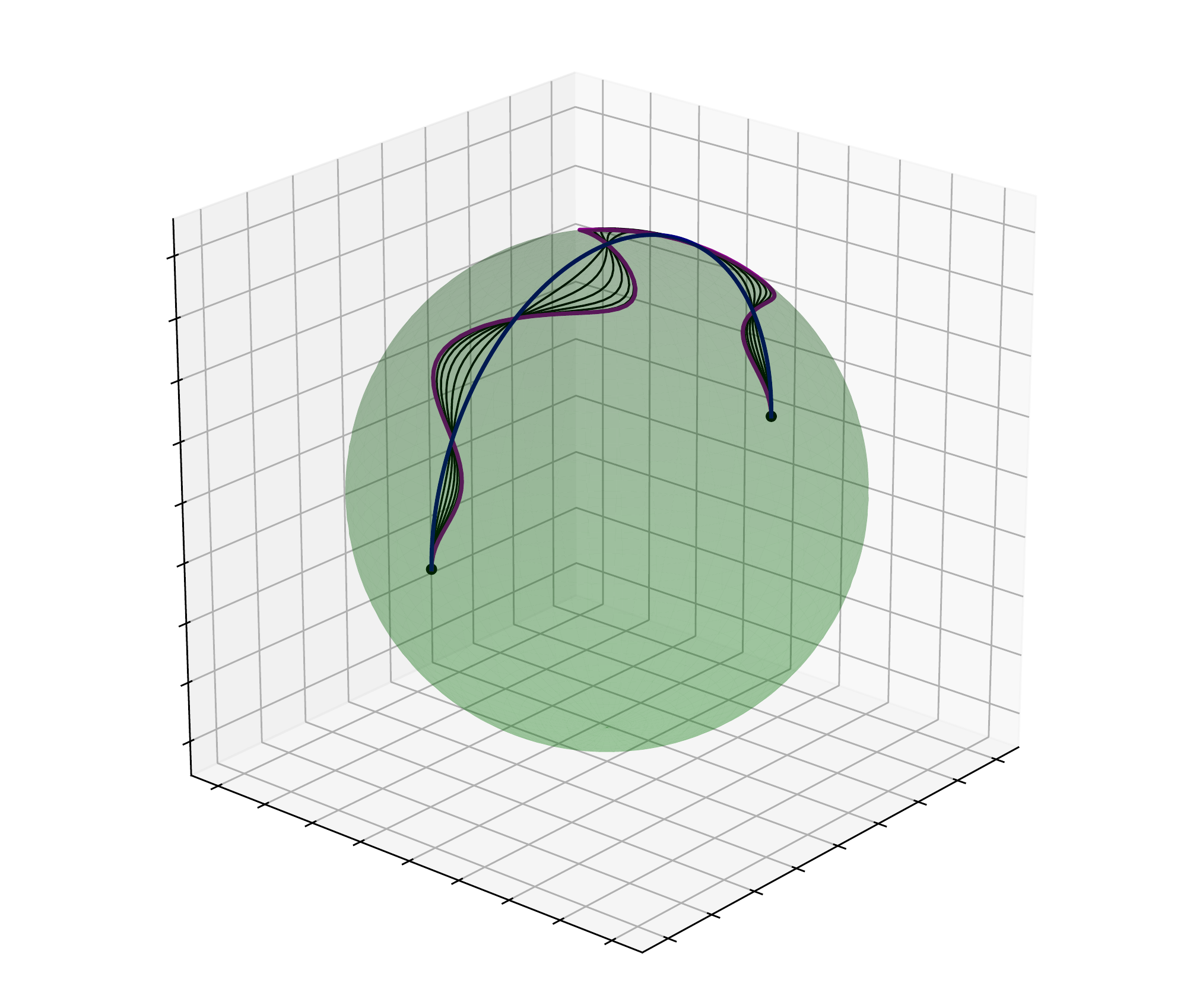}
	\vspace{-12pt}
	\caption{Stabilizing the wobbling 360 shove-it}	
	\label{fig:wobblehomotopy}
	\vspace{-6pt}
\end{figure}
%

\section{Conclusion}
\label{sec:conclusion}

We have ongoing work regarding the deformations of strictly physical skateboard tricks. This involves an analysis of the deformation of closed geodesics of the associated Riemannian metric on the Lie group $SO(3)$. 

For a rigid body with enough symmetries this is a simpler task, when the associated Riemannian metric is bi-invariant its geodesics correspond to 1-parameter subgroups of $SO(3)$, which can be studied through its Lie algebra. However the asymmetry of the skateboard makes it so that the relevant Riemannian metric in this case is not bi-invariant, so that a more careful analysis of the closed geodesics becomes more difficult.

Finally we have developed a set of Python scripts which can generate rotation matrices and animations for a number of skateboard tricks. You can find them in the following GitHub repository:

\begin{center}
    \href{https://github.com/holomorpheus/topological-flips}{https://github.com/holomorpheus/topological-flips}
\end{center}


\textbf{Acknowledgements}\vspace{4pt}

We would like to thank Jos\'{e} Bravo for his contributions to the animation scripts and Lance Johnson, Garrett Lent, Brenden Mucklow and Diego Vera for our discussions on algebraic topology.



\begin{thebibliography}{1}

\bibitem{AM}Abraham R.; Marsden J. E.:
{Foundations of Mechanics}, American Mathematical Society, 2nd edition, (2008)

\bibitem{Ar}Arnold, V.:
{Sur la g\'{e}om\'{e}trie diff\'{e}rentielle des groupes de Lie de dimension infinie et ses applications \`{a} l'hydrodynamique des fluides parfaits}, Annales de l'Institut Fourier \textbf{16.1}, (1966)

\bibitem{Ar2}Arnold, V.:
{Mathematical Methods of Classical Mechanics}, Springer, (1989)

\bibitem{Ha}Hatcher, A.:
{Algebraic Topology}, Cambridge Univ. Press, (2001)

\bibitem{HSS}Holm, D. D.; Schmah T.; Stoica C.:
{Geometric Mechanics and Symmetry: From Finite to Infinite Dimensions}, Oxford University Press (2009)

\end{thebibliography}
\end{document}